\documentclass[11pt]{article}
\usepackage[numbers,sort&compress]{natbib}
\usepackage{appendix}
\usepackage{enumerate}
\usepackage{amscd}
\usepackage{amsmath}
\usepackage{latexsym}
\usepackage{amsfonts}
\usepackage{amssymb}
\usepackage{amsthm}
\usepackage{verbatim}
\usepackage{mathrsfs}
\usepackage{enumerate}
\usepackage{hyperref}

\theoremstyle{plain}
\theoremstyle{definition}\newtheorem{theorem}{Theorem}[section]
\theoremstyle{plain}
\theoremstyle{plain}\newtheorem{coro}[theorem]{Corollary}
\theoremstyle{plain}
\theoremstyle{remark}\newtheorem{remark}{Remark}[section]
\usepackage{xcolor}

\newcommand{\Div}{\mathrm{div}\,}
\newcommand{\B}{\Big}

\newcommand{\be}{\begin{equation}}
\newcommand{\ee}{\end{equation}}
 \newcommand{\ba}{\begin{aligned}}
 \newcommand{\ea}{\end{aligned}}

  \newcommand{\f}{\frac}
    
  \newcommand{\ben}{\begin{enumerate}}
   \newcommand{\een}{\end{enumerate}}

\newcommand{\Rmnum}[1]{\expandafter\@slowromancap\romannumeral #1@}

\allowdisplaybreaks

\numberwithin{equation}{section}
\begin{document}
%%%%%%%%%%%%%%%%%%%%%%%%%%%%%%%%%%%%%%%%%%%%%%%%%%%%%%%%%%%%%%%%%%%%%%%%%%%%%%%%%%%%%%%%%%%%%%%%%%%%
\title{Four-thirds law of energy and magnetic helicity  in electron and Hall
magnetohydrodynamic fluids}
\author{Yanqing Wang\footnote{   College of Mathematics and   Information Science, Zhengzhou University of Light Industry, Zhengzhou, Henan  450002,  P. R. China Email: wangyanqing20056@gmail.com}   ~  %\,
      and  Otto Chkhetiani\footnote{A. M. Obukhov Institute of Atmospheric Physics, Russian Academy of Sciences, Pyzhevsky per. 3, Moscow, 119017, Russia.  Email: ochkheti@ifaran.ru}  }
\date{}
\maketitle
\begin{abstract}
 In this paper, by exploiting the feature of the Hall term, we
   establish
   some local version  four-thirds   laws   for the dissipation rates of energy and  magnetic helicity in both electron and Hall
magnetohydrodynamic equations in the sense of  Duchon-Robert type.
New $4/3$ laws for  the dissipation rates of    magnetic helicity in these systems
are first observed and  four-thirds  law   involving  the dissipation rates of energy for the Hall
magnetohydrodynamic equations generalizes   the work of Galtier.

  \end{abstract}
\section{Introduction}
\label{intro}
\setcounter{section}{1}\setcounter{equation}{0}

The energy distribution among scales and the energy flux in turbulence can be given in terms of third-order structure function in configuration space (see e.g. \cite{[AB],[Frisch],[MY]}). Two known exact relations for the third-order structural function in an incompressible fluid are Kolmogorov's $4/5$ law for longitudinal velocity pulsations in \cite{[Kolmogorov]}and Yaglom's $4/3$  law for mixed moments of the velocity and temperature fields in  \cite{[Yaglom]}.

 There exist  a lot of  generalized Kolmogorov    and  Yaglom type laws  involving the  energy, cross-helicity and helicity in
  the incompressible Euler equations, the magnetohydrodynamic system and other turbulence models.
  They are the few rigorous results in the theory of turbulence  and are confirmed by numerical simulation(see e.g. \cite{[AOAZ],[AB],[Chkhetiani1],[Chkhetiani],
 [GPP],[PP1],[PP2],[Podesta],[PFS],[YRS],[WWY1],[MY],[Galtier1],[FGSMA],
 [HVLFM],[Galtier]}). As \cite{[Kolmogorov],[Yaglom]}, almost all deductions of these laws   rely on the corresponding K\'arm\'an-Howarth  equations. Without an application of the  K\'arm\'an-Howarth  equations,
 the following version of  four-thirds law and  four-fifths law obtained in \cite{[Eyink1],[DR]} reads
\begin{align}
 &S_{1}(v)=-\f43D_{1}(v), \label{YL}\\
 &S_{L}(v)=-\f45D(v),\label{KL}
\end{align}
where
$$\ba
&S(v)=-\lim\limits_{\lambda\rightarrow0}S (v,\lambda)=-
\lim\limits_{\lambda\rightarrow0}\f{1}{\lambda}\int_{\partial B }\ell\cdot\delta v(\lambda\ell)|\delta v(\lambda\ell)|^{2}\f{d\sigma(\ell) }{4\pi}, \\ &S_{L}(v)=-\lim\limits_{\lambda\rightarrow0}S_{L}(v,\lambda)=-
\lim\limits_{\lambda\rightarrow0}\f{1}{\lambda}\int_{\partial B } \ell \cdot\delta v(\lambda\ell) |\delta v_{L}(\lambda\ell)|^{2}\f{d\sigma(\ell) }{4\pi},
\ea$$
and
 \be\label{drKHMr}\ba
&D_{1}(v)=-\lim\limits_{\varepsilon\rightarrow0}\f14
\int_{\mathbb{T}^{3}}\nabla\varphi_{\varepsilon}(\ell)\cdot\delta v(\ell)|\delta v(\ell)|^{2}d\ell,\\&D(v)=-\lim\limits_{\varepsilon\rightarrow0}\f14
\int_{\mathbb{T}^{3}}\nabla\varphi_{\varepsilon}(\ell)\cdot\delta v(\ell)|\delta v_{L}(\ell)|^{2}d\ell,
\ea
\ee
here, $\sigma(x)$ stands for the surface measure on the sphere $\partial B=\{x\in \mathbb{R}^{3}: |x|=1\}$ and $\varphi$ is  some smooth non-negative function  supported in $\mathbb{T}^{3}$ with unit integral and $\varphi_{\varepsilon}(x)=\varepsilon^{-3}\varphi(\f{x}{\varepsilon})$.
$\delta v_{L}(r)=\delta v (r)\cdot\f{r}{|r|}=(v(x+r)-v(x))\cdot\f{r}{|r|}$ stands for the longitudinal velocity increment.
The dissipation term $\eqref{drKHMr}_{1}$ was initial   by Duchon-Robert in \cite{[DR]}. Very recently, in the spirit of \cite{[DR]},  the first four-thirds relation for the Oldroyd-B model and  six new 4/3 laws for the subgrid scale $\alpha$-models of turbulence were obtained in \cite{[WWY1]}. Moreover, in \cite{[WWY1]}, almost all 4/3 relation  in the temperature equation,  the inviscid MHD equations and the Euler equations   can be written in the form of \eqref{YL}.
The similarity of various  turbulence models in \cite{[WWY1]}
  is the nonlinear term in terms of convection type.  Besides the standard MHD equations and Leray-$\alpha$ MHD equations in \cite{[BT]}, the electronic (EMHD) and Hall (HMHD)
magnetohydrodynamic equations play  an important role in the theory of
plasma (see e.g. \cite{[Vainshtein73],[Galtier1],[Chkhetiani2],[Biskamp99],[KCY]}   and references therein).  Both the electronic (EMHD)
and Hall (HMHD) magnetohydrodynamic system enjoy the energy and  helicity conserved laws (see \cite{[Chkhetiani2],[Galtier1]}). The authors in \cite{[WWY1]} pointed out that the  dissipation term of conserved quantity as $\eqref{drKHMr}_{1}$  immediately a 4/3 relation.
Based on this, a natural question is whether there exist    four-thirds relations of energy and helicity in electronic and Hall
magnetohydrodynamic. The objective of this paper is to consider this issue.  Before we state the main results, we recall the following EMHD equation
\be\label{plasmahydrodynamics}
b_{t}+\text{d}_{\text{I}}\nabla\times[(\nabla\times b)\times b]=0,\text{div}b=0,
\ee
where $b$ represents   the magnetic field and $\text{d}_{\text{I}}$ stands for the ion inertial length. Without loss of generality, we set $\text{d}_{\text{I}}=1$.
Cascade processes in such a representation were considered in \cite{[Frick03]}. We formulate the result involving the EMHD equations as follows.
     \begin{theorem}\label{the1.1}  Let
 $b$ be a weak solution of the EMHD equations
 \eqref{plasmahydrodynamics} and the electric current $\vec{j}=\nabla\times b$.  Assume that for any $1<p,q,m,n<\infty$ with  $\f2p+\f1m=1,\f2q+\f1n=1 $ such that $(b, v)$ satisfies
 \be\label{the1.1c}
b\in L^{\infty}(0,T;L^{2}(\mathbb{T}^{3}))\cap L^{p}(0,T;L^{q}(\mathbb{T}^{3}))\ \text{and} ~~\vec{j}\in L^{m}(0,T;L^{n}(\mathbb{T}^{3})). \ee
  Then the function
\be\label{drdissipationterm}
D(b,\vec{j};\varepsilon)= \f18\int_{\mathbb{T}^{3}}\nabla\varphi_{\varepsilon}(\ell)\cdot\delta
\vec{ j}(\ell) |\delta  b(\ell)|^{2} d\ell-\f14\int_{\mathbb{T}^{3}}\nabla\varphi_{\varepsilon}(\ell)\cdot\delta b(\ell) | \delta  \vec{j}(\ell)\cdot\delta b(\ell)|  d\ell,\ee  converges to a distribution $D(b,\vec{j})$ in the sense of distributions as $\varepsilon\rightarrow0$, and $D(b,\vec{j})$ satisfies the local equation of energy
 $$ \partial_{t}(\f12|b|^{2})  +\f12\text{div}( [\text{div}(b\otimes b)\times b]
-\f14\text{div} (\vec{ j}|b|^{2})+ \f12\text{div}(b\vec{ j}\cdot b)=D(b,\vec{j}),
$$
in the sense of distributions.
Moreover, there holds the following $4/3$ law
\be\label{EMHDYaglom4/3law}
-\f12S_{1}(\vec{j},b,b)+S_{2}(b,\vec{j},b)=-\f43D(v,\theta),
\ee
where $$\ba
&S_{1}(\vec{j},b,b)=-\lim\limits_{\lambda\rightarrow0}S_{1} ( \vec{j},b,b;\lambda)=-
\lim\limits_{\lambda\rightarrow0}\f{1}{\lambda}\int_{\partial B } \ell \cdot\delta \vec{j} (\lambda\ell)|\delta b(\lambda\ell)|^{2}\f{d\sigma(\ell) }{4\pi},
\\
&S_{2}(b,\vec{j},b)=-\lim\limits_{\lambda\rightarrow0}S _{2}(b,\vec{j}b;\lambda)=-
\lim\limits_{\lambda\rightarrow0}\f{1}{\lambda}\int_{\partial B } \ell \cdot\delta b (\lambda\ell)|\delta \vec{j}(\lambda\ell)\cdot\delta b(\ell)| \f{d\sigma(\ell) }{4\pi}.
\ea$$
 \end{theorem}
 \begin{remark}
 A kind local equation of energy  for the   Hall-MHD equations \eqref{hallMHD}
  with the dissipation term $D(\vec{j},\vec{j},b,\varepsilon)=\f12\int\varphi_{\varepsilon}(\ell)[\delta\vec{ j}\cdot\delta(\vec{j}\times b)]$
 was derived by Galtier in \cite{[Galtier]} and its   four-three law  can be found in \cite{[Galtier1]}. EMHD system \eqref{plasmahydrodynamics}   can be viewed as a sub-system of the Hall-MHD equations \eqref{hallMHD}, therefore,  this theorem generalizes the corresponding results \cite{[Galtier],[Galtier1]}.
 \end{remark}
  \begin{remark}
 It is worth remarking that the  dissipation term \eqref{drKHMr}  for the energy in the EMHD is similar to the one for the helicity in the Euler equations in \cite{[WWY1]}. Meanwhile, the structure of  dissipation term \eqref{drdissipationterm2}  for the  magnetic helicity in the EMHD is the same as the one for the energy  in the Euler equations in \cite{[DR]}.
\end{remark}
 \begin{remark}
 It is shown that the helicity is conserved provided that  $v\in L^3(0,T;B^{\f23}_{3,q^{\natural}})$ with $q^{\natural}<\infty$ in \cite{[DKL]}. Hence, if $m\geq3$ in \eqref{the1.1c}, we require $n<9/4$ in this theorem. Since a special case of \eqref{the1.1c} is $p=m=3$, $q=\f92$ and $n=\f95$, the condition \eqref{the1.1c} in no empty.
\end{remark}
Compared with nonlinear term  in terms of convection type of  the models in \cite{[DR],[WWY1]},    the Hall term $\nabla\times[(\nabla\times b)\times b]$ in    the EMHD and HMHD equations
involves the second order derivative rather than    the first order derivative.  To establish \eqref{EMHDYaglom4/3law}, a natural strategy is to reformulate  the the Hall term $\nabla\times[(\nabla\times b)\times b]$ as a convection type to apply the following equations
$$ b_{t}+  \text{div}(b\otimes\vec{j})- \text{div}(\vec{j}\otimes b)=0.$$
However, the EMHD equations in this form  still do not match the dissipation term  \eqref{drdissipationterm} directly. Precisely, the left hand side of \eqref{2.7} is lack of the term  $[\text{div}(b\otimes b)]^{\varepsilon}\cdot \vec{j} +[\text{div}(b\otimes b)]\cdot \vec{j}^{\varepsilon} $. Fortunately,
when we study the 4/3 laws for the magnetic helicity in this system, we observe that if we replace the magnetic vector potential $A$ in
\eqref{h1}  and \eqref{h2} by $B$, we immediately derive this desired term, which  inspires us to use the  following equivalent form of  the EMHD
$$ b_{t}+\nabla\times[\text{div}(b\otimes b)]=0.$$ Based on this, we get the critical equation \eqref{2.15}, which  is appropriate for the dissipation term  \eqref{drdissipationterm}.  Indeed, we will provide two slightly different methods  to obtain \eqref{2.15}.  This together with technique used in \cite{[DR],[WWY1]} help us to
achieve the desired relation \eqref{EMHDYaglom4/3law}.

As \cite{[DR],[WWY1]}, we apply the dissipation term \eqref{drdissipationterm}
to establish new  sufficient condition for implying  magnetic helicity conservation  of weak solutions of EMDH equations
 \eqref{plasmahydrodynamics}.
 \begin{coro}\label{coro1.6}
 We use the notations in Theorem \ref{the1.1}.
Assume that $b$ and $\vec{ j}$ satisfy
\be\label{1.23}
\ba
&\B(\int_{\mathbb{T}^{3}}|b(x+\ell,t)-b(x,t)|^{\f{9}{2}}dx\B)^{\f{2}{9}}\leq C(t)^{\f{1}{r_{1}}}|\ell|^{\alpha}\sigma^{\f13}(\ell),\\ &\B(\int_{\mathbb{T}^{3}}|\vec{ j}(x+\ell,t)-\vec{ j}(x,t)|^{\f95}dx\B)^{\f59}\leq C(t)^{\f{1}{r_{2}}}|\ell|^{\beta}\sigma^{\f13}(\ell),\\
&\text{with}\ \f{2}{r_{1}}+\f{1}{r_{2}}=1, 1<r_{1},r_{2}<\infty, 2\alpha+\beta\geq1,
\ea\ee
where both of $C_{i}(t)$ for $i=1,2$  are integrable functions on $[0,T]$, and $\sigma_{i}(\ell)$ for $i=1,2$  are both bounded functions on some neighborhood of the origin. Suppose that at least one of  $\sigma_{i}(\ell)$ obeys $\sigma_{i}(\ell) \rightarrow0$ as $\ell \rightarrow0$. Then the energy is conserved.
 \end{coro}
\begin{remark}
Corollary \ref{coro1.6} implies that $ b\in L^{r_{1}}(0,T; B^{\alpha}_{\f92,\infty})$ and
$ \vec{ j}\in L^{r_{1}}(0,T; B^{\beta}_{\f95,\infty})$ with $2\alpha+\beta>1$ and $\f{2}{r_{1}}+\f{1}{r_{2}}=1$ guarantee that the energy of weak solutions of the EMHD is invariant. This is close to the helicity conservation criterion proved by Chae in \cite{[Chae1]}. %For other , the reader is refer to work \cite{[]}
  \end{remark}
Next, we consider the second conserved quantity
magnetic helicity
\be\label{magnetic helicity}
\int_{\mathbb{T}^d}  A\cdot {\rm curl\,} A\  dx,
\ee
 as a
topological invariant of the motion of plasma,
where $ A={\rm curl}^{-1}b$ stands for   the magnetic vector potential.  From EMHD equations \eqref{plasmahydrodynamics},  we deduce the magnetic vector potential equations
\be\label{magneticpotentialeq}
A_{t}+ (\nabla\times b)\times b   +\nabla \pi=0, \text{div}A=0.
\ee

\begin{theorem}\label{the1.2}
 Let
 $b$ be a weak solution of EMDH equations
 \eqref{plasmahydrodynamics} and magnetic vector potential $A$  satisfy \eqref{magneticpotentialeq}.  Assume that
 \be\label{the1.2c}
\vec{j} \in L^{\infty}(0,T;L^{\f32}(\mathbb{T}^{3})) \ \text{and}\ A\in C((0,T)\times\mathbb{T}^{3}). \ee
  Then the function
\be \label{drdissipationterm2}D_{mh}(b,\varepsilon )= -\f12\int_{\mathbb{T}^{3}}\nabla\varphi_{\varepsilon}(\ell)\cdot\delta b(\ell) | \delta b(\ell)|^{2}  d\ell,\ee  converges to a distribution $D_{mh}(b)$ in the sense of distributions as $\varepsilon\rightarrow0$, and $D_{mh}(b)$ satisfies the local energy balance
 $$ \ba
& \partial_{t}(bA )
 + \text{div}([\text{div}(b\otimes b)]\times A )   + \text{div}[ \pi b ]+ \text{div}(b|b|^{2})=D_{mh}(b)
\ea$$
in the sense of distributions.
Moreover, there holds the following $4/3$ law
\be\label{EMHDMHYaglom4/3law}
 S (b,b,b)=-\f43D_{mh}(b),
\ee
where $$\ba
&S (b,b,b)=-\lim\limits_{\lambda\rightarrow0}S  ( b,b,b;\lambda)=-
\lim\limits_{\lambda\rightarrow0}\f{1}{\lambda}\int_{\partial B } \ell \cdot\delta b (\lambda\ell)|\delta b(\lambda\ell)|^{2}\f{d\sigma(\ell) }{4\pi}.
 \ea$$
 \end{theorem}
\begin{remark}
Unlike 4/3 law  \eqref{EMHDYaglom4/3law} for the energy,
to the knowledge of the authors, the  four-thirds relationship \eqref{EMHDMHYaglom4/3law} of magnetic   helicity in   system \eqref{plasmahydrodynamics} is completely new. It is an interesting question to derive \eqref{EMHDMHYaglom4/3law} via the K\'arm\'an-Howarth  equations.
\end{remark}
\begin{remark}
As Corollary \eqref{coro1.6}, the dissipation term \eqref{drdissipationterm2} means that weak solutions of the EMHD  preserve the magnetic helicity  if  $ b\in L^{r_{1}}(0,T; B^{\alpha}_{3,\infty})$  with $\alpha >1/3$ .
\end{remark}
 We turn our attention to the following Hall MHD equations
 \be\left\{\ba\label{hallMHD}
&u_{t}+u\cdot\nabla u-b\cdot\nabla b+\nabla\Pi =0, \\
&b_{t}+u\cdot\nabla b-b\cdot\nabla u+\nabla\times[(\nabla\times b)\times b] =0, \\
&\Div u=\Div b=0,
 \ea\right.\ee
 where $v$ represents the velocity field of the flow and
 $\Pi$ stands for  the  pressure of the fluid, respectively.
The next goal is to extend the   four-thirds  law of energy and helicity  from the  electron magnetohydrodynamic  system \eqref{plasmahydrodynamics} to the  Hall
magnetohydrodynamic equations \eqref{hallMHD}.

\begin{theorem}\label{the1.4}
 Let  the pair
 $(u, b)$ be a weak solution of HMHD equations
 \eqref{hallMHD}.  Assume that for any $1<p,q,m,n<\infty$ with  $\f2p+\f1m=1,\f2q+\f1n=1 $ such that $(\theta, v)$ satisfies
 \be\label{the1.3c}\ba
& u  \in L^{\infty}(0,T;L^{2}(\mathbb{T}^{3}))\cap L^{3}(0,T;L^{3}(\mathbb{T}^{3})),\\
& b\in L^{\infty}(0,T;L^{2}(\mathbb{T}^{3}))\cap L^{p}(0,T;L^{q}(\mathbb{T}^{3}))\ \text{and} ~~\vec{j}\in L^{m}(0,T;L^{n}(\mathbb{T}^{3})).
 \ea\ee
  Then the function
$$\ba
D(u,b,\vec{j};\varepsilon)=& -\f14\int_{\mathbb{T}^{3}}\nabla\varphi_{\varepsilon}(\ell)\cdot\delta u(\ell)|\delta  u(\ell)|^{2}d\ell -\f14\int_{\mathbb{T}^{3}}\nabla\varphi_{\varepsilon}(\ell)\cdot\delta b(\ell) | \delta  u(\ell)\cdot\delta b(\ell)|  d\ell\\&+\f12 \int_{\mathbb{T}^{3}}\nabla\varphi_{\varepsilon}(\ell)\cdot\delta u(\ell)|\delta  b(\ell)\cdot\delta  b(\ell)| d\ell\\&+ \f18\int_{\mathbb{T}^{3}}\nabla\varphi_{\varepsilon}(\ell)\cdot\delta\vec{ j}(\ell) |\delta  b(\ell)|^{2} d\ell-\f14\int_{\mathbb{T}^{3}}\nabla\varphi_{\varepsilon}(\ell)\cdot\delta b(\ell) | \delta  \vec{j}(\ell)\cdot\delta b(\ell)|  d\ell
\ea$$  converges to a distribution $D(u,b,\vec{j})$ in the sense of distributions as $\varepsilon\rightarrow0$, and $D(u,b,\vec{j})$ satisfies the local equation of energy
 $$\ba &\partial_{t}(\f{u^{2}+b^{2}}{2}  )  +\text{div}\B[u\B( \f12(|u|^{2}+|b |^{2})+\Pi\B)-b(b\cdot u)\B] \\& +\f12\text{div}( [\text{div}(b\otimes b)\times b]
-\f14\text{div} (\vec{ j}|b|^{2})+ \f12\text{div}(b\vec{ j}\cdot b)=D(u,b,\vec{j}),
\ea$$
in the sense of distributions.
Moreover, there holds the following $4/3$ law
\be\label{HMHDYaglom4/3law}
 S_{3}(u,u,u)+S_{4}(u,b,b)-2 S_{5}(b,u,b )-\f12S_{1}(\vec{j},b,b)+S_{2}(b,\vec{j},b)=-\f43D (u,b,\vec{j}),
\ee
where
$$\ba
&S_{3}(u,u,u)=-\lim\limits_{\lambda\rightarrow0}S_{1} ( u,u,u;\lambda)=-
\lim\limits_{\lambda\rightarrow0}\f{1}{\lambda}\int_{\partial B } \ell \cdot\delta u(\lambda\ell)|\delta u(\lambda\ell)|^{2}\f{d\sigma(\ell) }{4\pi},\\
&S_{4}(u,b,b)=-\lim\limits_{\lambda\rightarrow0}S_{2} ( u,b,b;\lambda)=-
\lim\limits_{\lambda\rightarrow0}\f{1}{\lambda}\int_{\partial B } \ell \cdot\delta u (\lambda\ell)|\delta b(\lambda\ell)|^{2}\f{d\sigma(\ell) }{4\pi},
\\
&S_{5}(b,u,b)=-\lim\limits_{\lambda\rightarrow0}S _{3}(b,u,b;\lambda)=-
\lim\limits_{\lambda\rightarrow0}\f{1}{\lambda}\int_{\partial B } \ell \cdot\delta b (\lambda\ell)|\delta u(\lambda\ell)\cdot\delta b(\ell)| \f{d\sigma(\ell) }{4\pi}.
\ea$$
 \end{theorem}
 \begin{remark}
 We would like to point out that
  that  the relation   \eqref{HMHDYaglom4/3law} is consistent with  the result  proved in \cite{ [HVLFM],[FGSMA]}.
 \end{remark}
 Besides   total energy conservation,  the smooth solution of the Hall MHD equations
 \eqref{hallMHD} obeys magnetic helicity
 conservation. For the HMHD equations \eqref{hallMHD}, we derive  from the following magnetic vector
potential equations and  $\eqref{VI}_{2}$ that
 \be\label{hmpotentialeq}
A_{t}-u\times b+ (\nabla\times b)\times b   +\nabla \pi=0, \text{div}A=0.
\ee
 There is little literature concerning  investigation of four-thirds law of helicity in  the  Hall
magnetohydrodynamic \eqref{hallMHD}. The final result is stated as follows.
\begin{theorem}\label{the1.5}
 Let
 $b$ be a weak solution of the HMDH equations
 \eqref{hallMHD} and magnetic vector potential $A$  satisfy \eqref{hmpotentialeq}.  Assume that
 \be\label{the1.4c}
\vec{j} \in L^{\infty}(0,T;L^{\f32}(\mathbb{T}^{3})), \ u\in L^{3}(0,T;L^{3}(\mathbb{T}^{3})) \text{and}\ A\in C((0,T)\times\mathbb{T}^{3}). \ee
  Then the function
$$D_{mh}(b,\varepsilon )= -\f12\int_{\mathbb{T}^{3}}\nabla\varphi_{\varepsilon}(\ell)\cdot\delta b(\ell) | \delta b(\ell)|^{2}  d\ell,$$  converges to a distribution $D_{mh}(b)$ in the sense of distributions as $\varepsilon\rightarrow0$, and $D_{mh}(b)$ satisfies the local equation of energy
 $$ \ba
& \partial_{t}(bA )
 + \text{div}([\text{div}(b\otimes b)]\times A )   + \text{div}[ \pi b ]+ \text{div}(b|b|^{2})=D_{mh}(b)
\ea$$
in the sense of distributions.
Moreover, there holds the following $4/3$ law
\be\label{HMHDMHYaglom4/3law}
 S (b,b,b)=-\f43D_{mh}(b,\varepsilon ),
\ee
where $$\ba
&S (b,b,b)=-\lim\limits_{\lambda\rightarrow0}S  ( b,b,b;\lambda)=-
\lim\limits_{\lambda\rightarrow0}\f{1}{\lambda}\int_{\partial B } \ell \cdot\delta b (\lambda\ell)|\delta b(\lambda\ell)|^{2}\f{d\sigma(\ell) }{4\pi}.
 \ea$$
 \end{theorem}
 \begin{remark}
 It seems that the relation \eqref{HMHDMHYaglom4/3law} is the first 4/3 law of   magnetic helicity in    Hall
magnetohydrodynamic equations. The reader may refer to \cite{[BG]} for other exact relations for the magnetic helicity in HMHD equaitons.
 \end{remark}
To end  this section, we introduce some notations which will be used in this paper.
Firstly, for $p\in [1,\,\infty]$, the notation $L^{p}(0,\,T;X)$ stands for the set of measurable functions $f$ on the interval $(0,\,T)$ with values in $X$ and $\|f\|_{X}$ belonging to $L^{p}(0,\,T)$.
Secondly, we will use the standard mollifier kernel, i.e. $\varphi(x)=C_0e^{-\frac{1}{1-|x|^2}}$ for $|x|<1$ and $\varphi(x)=0$ for $|x|\geq 1$, where $C_0$ is a constant such that $\int_{\mathbb{R}^3}\varphi (x) dx=1$. Eventually, for $\varepsilon>0$, we denote the rescaled mollifier by  $\varphi_{\varepsilon}(x)=\frac{1}{\varepsilon^3}\varphi(\frac{x}{\varepsilon})$, and for any function $f\in L^1_{\textrm{loc}}(\mathbb{R}^3)$, its mollified version is defined by
$$
f^\varepsilon(x)=\int_{\mathbb{R}^3}\varphi_{\varepsilon}(x-y)f(y)dy,\ \ x\in \mathbb{R}^3.
$$
The paper is organized as follows.  Section 2 is concerned with exact relation of the energy and the magnetic helicity in the electronic magnetohydrodynamic equations.
In Section 3,  we  establish the  four-thirds laws in the  Hall
magnetohydrodynamic equations.  Finally,   concluding remarks are given in
section 4.
\section{Four-thirds  laws in  electron magnetohydrodynamic system}
This section is devoted to the study 4/3 laws for   the dissipation rates of  energy and magnetic helicity in
electron magnetohydrodynamic  equations \eqref{plasmahydrodynamics}. Before we begin the proof, we recall some vector identities as follows,
\be\ba\label{VI}
&\nabla(\vec{A}\cdot \vec{B})=\vec{A}\cdot\nabla \vec{B}+\vec{B}\cdot\nabla \vec{A}+\vec{A}\times\text{curl}\vec{B}+\vec{B}\times(\nabla\times\vec{A}),\\
& \nabla\times(\vec{A}\times \vec{B})=\vec{A} \text{div} \vec{B}-\vec{B} \text{div}  \vec{A}+\vec{B}\cdot\nabla \vec{A}-\vec{A}\cdot\nabla \vec{B},\\
&\vec{A}\cdot(\nabla\times \vec{B})=\text{div}(\vec{B}\times \vec{A})+\vec{B}\cdot(\nabla\times \vec{A}),
\ea\ee
which will be frequently used in this paper.
Combining this and the divgence-free condition $ \Div \vec{j}=0$, one obtains
 \be\label{identity}\ba
&b\cdot\nabla b = \f{1}{2}\nabla |b|^{2}+\vec{j}\times b,\\
&\nabla\times(\vec{j}\times b)=b\cdot\nabla \vec{j}-\vec{j}\cdot\nabla b,\ea\ee
which turns out that
\begin{align}\label{non1}
 &\nabla\times[(\nabla\times b)\times b]= \nabla\times[\vec{j}\times b]=
 \text{div}(b\otimes\vec{j})- \text{div}(\vec{j}\otimes b),\\
 \label{non2}&\nabla\times[(\nabla\times b)\times b]= \nabla\times[\vec{j}\times b]=\nabla\times[\text{div}(b\otimes b)-\nabla\f12|b|^{2}]=\nabla\times[\text{div}(b\otimes b)].
\end{align}
Hence, we get two equivalent forms of EMDH equation \eqref{plasmahydrodynamics}
\begin{align}\label{eq1}
& b_{t}+  \text{div}(b\otimes\vec{j})- \text{div}(\vec{j}\otimes b)=0,
\\
 \label{eq2}
& b_{t}+\nabla\times[\text{div}(b\otimes b)]=0.
\end{align}
\subsection{Exact relation of energy in the EMHD system}

 \begin{proof}[Proof of Theorem \ref{the1.1} ]
We conclude by mollifying the  equation \eqref{eq1} that
$$
b_{t}^{\varepsilon}+  \text{div}(b\otimes\vec{j})^{\varepsilon}- \text{div}(\vec{j}\otimes b)^{\varepsilon}=0.
$$
After multiplying the above equation by $b$ and the  equation  \eqref{eq1} by $b^{\varepsilon}$, respectively,  we derive from  summing them  together that
\be\label{2.7}
 \partial_{t}(bb^{\varepsilon})+ \text{div}(b\otimes\vec{j})^{\varepsilon} b+ \text{div}(b\otimes\vec{j})b^{\varepsilon}- \text{div}(\vec{j}\otimes b)^{\varepsilon}b- \text{div}(\vec{j}\otimes b)b^{\varepsilon}=0.
\ee
Likewise,
\be\label{2.8}
 \partial_{t}(bb^{\varepsilon})+b\cdot\{\nabla\times[\text{div}(b\otimes b)]^{\varepsilon}\}+b^{\varepsilon}\cdot\{\nabla\times[\text{div}(b\otimes b)]\}=0.
\ee
With the help of  identity $\eqref{VI}_{3}$, we know that
$$\ba
b\cdot\{\nabla\times[\text{div}(b\otimes b)]^{\varepsilon}\}=&
\text{div}([\text{div}(b\otimes b)]^{\varepsilon}\times b)+[\text{div}(b\otimes b)]^{\varepsilon}\cdot(\nabla\times b)\\=&
\text{div}([\text{div}(b\otimes b)]^{\varepsilon}\times b)+[\text{div}(b\otimes b)]^{\varepsilon}\cdot \vec{j}
\ea$$
and
$$b^{\varepsilon}\cdot\{\nabla\times[\text{div}(b\otimes b)]\}=\text{div}([\text{div}(b\otimes b)]\times b^{\varepsilon})+[\text{div}(b\otimes b)] \cdot \vec{j}^{\varepsilon}.$$
Inserting the latter two equations into \eqref{2.8}, we know that
\be\ba \label{2.9}
 &\partial_{t}(bb^{\varepsilon})+\text{div}([\text{div}(b\otimes b)]^{\varepsilon}\times b)+\text{div}([\text{div}(b\otimes b)]^{\varepsilon}\times b)\\&+[\text{div}(b\otimes b)]^{\varepsilon}\cdot \vec{j}+[\text{div}(b\otimes b)]\cdot \vec{j}^{\varepsilon}=0.
\ea\ee
Putting \eqref{2.7} and \eqref{2.9} together, we arrive at
 \be\ba \label{Hkey1}
\f12& \partial_{t}(bb^{\varepsilon})+\f14\{\text{div}([\text{div}(b\otimes b)]^{\varepsilon}\times b)+\text{div}([\text{div}(b\otimes b)]\times b^{\varepsilon})\\&+[\text{div}(b\otimes b)]^{\varepsilon}\cdot \vec{j}+[\text{div}(b\otimes b)]\cdot \vec{j}^{\varepsilon}+ \text{div}(b\otimes\vec{j})^{\varepsilon} b+ \text{div}(b\otimes\vec{j})b^{\varepsilon}\\&- \text{div}(\vec{j}\otimes b)^{\varepsilon}\cdot b- \text{div}(\vec{j}\otimes b)^{\varepsilon}\cdot b\}=0.
\ea\ee
It is easy to check that
$$\ba
\partial_{k}(\vec{j}_{k}b_{i})^{\varepsilon}b_{i}+\partial_{k}(\vec{j}_{k}b_{i})b_{i}^{\varepsilon}
=&\partial_{k}(\vec{j}_{k}b_{i}b_{i}^{\varepsilon})+
\partial_{k}(\vec{j}_{k}b_{i})^{\varepsilon}b_{i}-
(\vec{j}_{k}b_{i})\partial_{k}b_{i}^{\varepsilon}\\
=&\text{div}(\vec{j} b\cdot b^{\varepsilon})+\partial_{k}(\vec{j}_{k}b_{i})^{\varepsilon}b_{i}
-(\vec{j}_{k}b_{i})\partial_{k}b_{i}^{\varepsilon},
\ea$$
which means that
\be\label{nonm1} \text{div}(\vec{j}\otimes b)^{\varepsilon}\cdot b+\text{div}(\vec{j}\otimes b)\cdot b^{\varepsilon}=\text{div}(\vec{j} b\cdot b^{\varepsilon})+\partial_{k}(\vec{j}_{k}b_{i})^{\varepsilon}b_{i}
-(\vec{j}_{k}b_{i})\partial_{k}b_{i}^{\varepsilon}.
\ee
A straightforward computation yields that
\be\ba\label{2.12}
&\partial_{k}(b_{k}b_{i})^{\varepsilon}\vec{j}_{i}+\partial_{k}(b_{k}b_{i}) j_{i}^{\varepsilon}=\partial_{k}(b_{k}b_{i}j_{i}^{\varepsilon}) +\partial_{k}(b_{k}b_{i})^{\varepsilon}\vec{j}_{i}-(b_{k}b_{i}) \partial_{k}j_{i}^{\varepsilon},\\
&\partial_{k}(b_{k}\vec{j}_{i})^{\varepsilon}b_{i}+\partial_{k}(b_{k}j_{i}) b _{i}^{\varepsilon}=
\partial_{k}(b_{k}j_{i}b _{i}^{\varepsilon})+\partial_{k}(b_{k}\vec{j}_{i})^{\varepsilon}b_{i}-(b_{k}j_{i}) \partial_{k}b _{i}^{\varepsilon}.
\ea\ee
Notice that
\be\ba\label{2.13}
&[\text{div}(b\otimes b)]^{\varepsilon}\cdot \vec{j}+[\text{div}(b\otimes b)]^{\varepsilon}\cdot \vec{j}+ \text{div}(b\otimes\vec{j})^{\varepsilon} b+ \text{div}(b\otimes\vec{j})b^{\varepsilon}\\
=&
\partial_{k}(b_{k}b_{i})^{\varepsilon}\vec{j}_{i}+\partial_{k}(b_{k}b_{i}) j_{i}^{\varepsilon}+
\partial_{k}(b_{k}\vec{j}_{i})^{\varepsilon}b_{i}+\partial_{k}(b_{k}j_{i}) b _{i}^{\varepsilon}.
\ea\ee
Inserting \eqref{nonm1} into \eqref{2.13}, we write
\be\ba\label{2.14}
&[\text{div}(b\otimes b)]^{\varepsilon}\cdot \vec{j}+[\text{div}(b\otimes b)]^{\varepsilon}\cdot \vec{j}+ \text{div}(b\otimes\vec{j})^{\varepsilon} b+ \text{div}(b\otimes\vec{j})b^{\varepsilon}\\
=&\text{div}[b (b\cdot j^{\varepsilon})] +\partial_{k}(b_{k}b_{i})^{\varepsilon}\vec{j}_{i}-(b_{k}b_{i}) \partial_{k}j_{i}^{\varepsilon}
+\text{div}[b (b^{\varepsilon}\cdot j)] +\partial_{k}(b_{k}b_{i})^{\varepsilon}\vec{j}_{i}-(b_{k}b_{i}) \partial_{k}j_{i}^{\varepsilon}\\
=&\text{div}[b (b\cdot j^{\varepsilon})] +\text{div}[b (b^{\varepsilon}\cdot j)] +\partial_{k}(b_{k}b_{i})^{\varepsilon}\vec{j}_{i}-(b_{k}b_{i}) \partial_{k}j_{i}^{\varepsilon} +\partial_{k}(b_{k}\vec{j}_{i})^{\varepsilon}b_{i}-(b_{k}j_{i}) \partial_{k}b _{i}^{\varepsilon}
\ea\ee
Plugging \eqref{nonm1}  and \eqref{2.14} into \eqref{Hkey1}, we have
\be\ba\label{2.15}
\f12& \partial_{t}(bb^{\varepsilon})+\f14\text{div}([\text{div}(b\otimes b)]^{\varepsilon}\times b)+\f14\text{div}([\text{div}(b\otimes b)]\times b^{\varepsilon})\\&+\f14\text{div}[b (b\cdot j^{\varepsilon})] +\f14\text{div}[b (b^{\varepsilon}\cdot j)]- \f14\text{div}(\vec{j} b\cdot b^{\varepsilon})\\&=\f14[\partial_{k}(\vec{j}_{k}b_{i})^{\varepsilon}b_{i}
-(\vec{j}_{k}b_{i})\partial_{k}b_{i}^{\varepsilon}]\\&-\f14[\partial_{k}(b_{k}b_{i})^{\varepsilon}\vec{j}_{i}-(b_{k}b_{i}) \partial_{k}j_{i}^{\varepsilon} +\partial_{k}(b_{k}\vec{j}_{i})^{\varepsilon}b_{i}-(b_{k}j_{i}) \partial_{k}b _{i}^{\varepsilon}].
\ea\ee
Before going further, we set
$$\delta \vec{j}_{i}(\ell)=\vec{j}_{i}(x+\ell)-\vec{j}_{i}(x)=J_{i}-\vec{j}_{i} ~ \text{and}  ~\delta  b_{i}(\ell)=b_{i}(x+\ell)-b_{i}(x)=B_{i} -b_{i}. $$
We notice that
$$\ba
&\int_{\mathbb{T}^{3}}\nabla\varphi_{\varepsilon}(\ell)\cdot\delta \vec{j}(\ell)|\delta  b(\ell)\cdot\delta  b(\ell)| d\ell\\
=&\int_{\mathbb{T}^{3}} \partial_{l_{k}}\varphi_{\varepsilon}(\ell)[\vec{J}_{k}(x+\ell)-\vec{j}_{k}(x)]
[B_{i}(x+\ell)-b_{i}(x)]^{2} d\ell\\=&\int_{\mathbb{T}^{3}} \partial_{l_{k}}\varphi_{\varepsilon}(\ell)[\vec{J}_{k}  B_{i} ^{2}
-2\vec{J}_{k}B_{i} b_{i} +\vec{J}_{k}b^{2}_{i}-\vec{j}_{k}b_{i}^{2}-\vec{j}_{k}B_{i}^{2}+2\vec{j}_{k}B_{i}b_{i}^{2}]  d\ell.
\ea$$
In view of  changing variables,  we deduce that
$$\ba
\int_{\mathbb{T}^{3}} \partial_{l_{k}}\varphi_{\varepsilon}(\ell) J_{k} B_{i}^{2}
   d\ell=&\int_{\mathbb{T}^{3}} \partial_{l_{k}}\varphi_{\varepsilon}(\ell) \vec{j}_{k}(x+\ell)
  b_{i}^{2}(x+\ell) d\ell\\
   =&\int_{\mathbb{T}^{3}} \partial_{\eta_{k}}\varphi_{\varepsilon}(\eta-x) \vec{j}_{k}(\eta)
  b_{i}^{2}(\eta) d\eta\\
   =&-\int_{\mathbb{T}^{3}} \partial_{x_{k}}\varphi_{\varepsilon}(\eta-x) \vec{j}_{k}(\eta)
b_{i}^{2}(\eta) d\eta\\
   =&-\partial_{k}(\vec{j}_{k} b_{i}^{2}\ast\varphi_{\varepsilon} )\\
   =&-\partial_{k}(\vec{j}_{k}b_{i}^{2} )^{\varepsilon}.
\ea$$
Arguing in the same manner as in the above derivation, we discover that
$$\ba
&  \int_{\mathbb{T}^{3}} \partial_{l_{k}}\varphi_{\varepsilon}(\ell)[\vec{J}_{k}  B_{i} ^{2}
-2\vec{J}_{k}B_{i} b_{i} +\vec{J}_{k}b^{2}_{i}-\vec{j}_{k}B_{i}^{2}+2\vec{j}_{k}B_{i}b_{i}^{2}-\vec{j}_{k}b_{i}^{2}]  d\ell\\
=&-\partial_{k}(\vec{j}_{k} b_{i}^{2} )^{\varepsilon}+2\partial_{k}(\vec{j}_{k} \theta  )^{\varepsilon} b_{i}-\partial_{k}\vec{j}_{k}^{\varepsilon} b_{i}^{2}+\vec{j}_{k}\partial_{k}( b_{i}^{2} )^{\varepsilon}-2\vec{j}_{k}\partial_{k}  b_{i} ^{\varepsilon}b_{i}-\vec{j}_{k}b_{i}^{2}\int_{\mathbb{T}^{3}} \partial_{l_{k}}\varphi_{\varepsilon}(\ell)   d\ell\\
=&\partial_{k}\B(\vec{j}_{k}( b_{i}^{2} )^{\varepsilon}-(\vec{j}_{k}b_{i}^{2} )^{\varepsilon}\B)+2\partial_{k}(\vec{j}_{k} b_{i}  )^{\varepsilon} b_{i}-2\vec{j}_{k}\partial_{k}  b_{i} ^{\varepsilon}b_{i},
\ea$$
which follows from that
\be\ba\label{2.26}
&\int_{\mathbb{T}^{3}}\nabla\varphi_{\varepsilon}(\ell)\cdot\delta \vec{j}(\ell)|\delta  b(\ell)\cdot\delta  b(\ell)| d\ell
= \partial_{k}\B(\vec{j}_{k}( b_{i}^{2} )^{\varepsilon}-(\vec{j}_{k}b_{i}^{2} )^{\varepsilon}\B)+2\partial_{k}(\vec{j}_{k} b_{i}  )^{\varepsilon} b_{i}-2\vec{j}_{k}\partial_{k}  b_{i} ^{\varepsilon}b_{i}.\ea\ee
Repeating the above deduction process, we derive from the divergence-free conditions $\Div v=0$ and $\Div b=0$,that
\be\label{2.17}\ba
&\int_{\mathbb{T}^{3}}\nabla\varphi_{\varepsilon}(\ell)\cdot\delta b(\ell) | \delta  j(\ell)\cdot\delta b(\ell)|  d\ell\\
=&-\partial_{k}(b_{k}\vec{j}_{i}b_{i})^{\varepsilon}
+\partial_{k}(b_{k}\vec{j}_{i})^{\varepsilon}b_{i}
+\partial_{k}(b_{k}b_{i})^{\varepsilon}\vec{j}_{i}-\partial_{k}b_{k}^{\varepsilon}b_{i}\vec{j}_{i}
+b_{k}\partial_{k}(\vec{j}_{i}b_{i})^{\varepsilon}-b_{k}b_{i}\partial_{k}\vec{j}_{i}^{\varepsilon}
-b_{k}\vec{j}_{i}\partial_{k}b_{i}^{\varepsilon}\\
=&-\partial_{k}(b_{k}\vec{j}_{i}b_{i})^{\varepsilon}
+\partial_{k}(b_{k}\vec{j}_{i})^{\varepsilon}b_{i}
+\partial_{k}(b_{k}b_{i})^{\varepsilon}\vec{j}_{i}
+b_{k}\partial_{k}(\vec{j}_{i}b_{i})^{\varepsilon}-b_{k}b_{i}\partial_{k}\vec{j}_{i}^{\varepsilon}
-b_{k}\vec{j}_{i}\partial_{k}b_{i}^{\varepsilon}\\
=&\partial_{k}\B(b_{k} (\vec{j}_{i}v_{i})^{\varepsilon}-(b_{k}\vec{j}_{i}b_{i})^{\varepsilon}\B)
+\partial_{k}(b_{k}b_{i})^{\varepsilon}\vec{j}_{i}
-b_{k}b_{i}\partial_{k}\vec{j}_{i}^{\varepsilon}
 +\partial_{k}(b_{k}\vec{j}_{i})^{\varepsilon}b_{i}
-b_{k}\vec{j}_{i}\partial_{k}b_{i}^{\varepsilon}.\ea\ee
The condition \eqref{the1.1c} ensures that the first term on the right hand side of both \eqref{2.26} and  \eqref{2.17}  converges to $0$ in the sense of distributions on $(0,T)\times\mathbb{T}^{3}$ as $\varepsilon\rightarrow0$. Consequently, the limit of
$$D(v,\vec{j};\varepsilon)= \f18\int_{\mathbb{T}^{3}}\nabla\varphi_{\varepsilon}(\ell)\cdot\delta\vec{ j}(\ell) |\delta  b(\ell)|^{2} d\ell-\f14\int_{\mathbb{T}^{3}}\nabla\varphi_{\varepsilon}(\ell)\cdot\delta b(\ell) | \delta  j(\ell)\cdot\delta b(\ell)|  d\ell,$$
is the same  as
$$\f14[\partial_{k}(\vec{j}_{k}b_{i})^{\varepsilon}b_{i}
-(\vec{j}_{k}b_{i})\partial_{k}b_{i}^{\varepsilon}] -\f14[\partial_{k}(b_{k}b_{i})^{\varepsilon}\vec{j}_{i}-(b_{k}b_{i}) \partial_{k}j_{i}^{\varepsilon} +\partial_{k}(b_{k}\vec{j}_{i})^{\varepsilon}b_{i}-(b_{k}j_{i}) \partial_{k}b _{i}^{\varepsilon}].$$
It remains to pass to the limit of terms on the left hand side of   \eqref{2.15}. Indeed, making use of \eqref{the1.1c} again, we know that
$ \f14\text{div}[b (b\cdot j^{\varepsilon})] +\f14\text{div}[b (b^{\varepsilon}\cdot j)]- \f14\text{div}(\vec{j} b\cdot b^{\varepsilon})$ tends to $ \f12\text{div}[b (b\cdot j )] - \f14\text{div}(\vec{j} b\cdot b )$    in the sense of distributions on $(0,T)\times\mathbb{T}^{3}$ as $\varepsilon\rightarrow0$. In view of the well-known Biot-Savart law, we deduce from $ \vec{j}\in L^{m}(0,T;L^{n}(\mathbb{T}^{3}))$ that $ \nabla b\in L^{m}(0,T;L^{n}(\mathbb{T}^{3}))$.  Therefore, we assert that $\f14\text{div}([\text{div}(b\otimes b)]^{\varepsilon}\times b)+\f14\text{div}([\text{div}(b\otimes b)]\times b^{\varepsilon})$
 converges to  $\f14\text{div}([\text{div}(b\otimes b)]^{\varepsilon}\times b)$ in the sense of distributions on $(0,T)\times\mathbb{T}^{3}$ as $\varepsilon\rightarrow0. $
 As consequence,
 the proof of the first part of of Theorem \ref{the1.1} is completed.
 The rest part is  devoted to  establishing \eqref{EMHDYaglom4/3law}.
Taking advantage of  the polar coordinates and changing variables several times, we end up with
\be\ba\label{key2}
&D(b,\vec{j};\varepsilon)\\=& \f18\int_{\mathbb{T}^{3}}\nabla\varphi_{\varepsilon}(\ell)\cdot\delta\vec{ j}(\ell) |\delta  b(\ell)|^{2} d\ell-\f14\int_{\mathbb{T}^{3}}\nabla\varphi_{\varepsilon}(\ell)\cdot\delta b(\ell) | \delta  \vec{j}(\ell)\cdot\delta b(\ell)|  d\ell\\
 =& \f18\int_{0}^{\infty}\int_{\partial B  } \f{r^{2}}{\varepsilon }\varphi'( |\zeta r| )\f{\zeta}{|\zeta|}\cdot[\vec{ j}(x+ \zeta r\varepsilon)-\vec{ j}(x)]
[b(x+\zeta r\varepsilon)-b(x)]^{2}d\sigma(\zeta)dr\\
&-\f14\int_{0}^{\infty}\int_{\partial B  } \f{r^{2}}{\varepsilon }\varphi'( |\zeta r| )\f{\zeta}{|\zeta|}\cdot[b(x+ \zeta r\varepsilon)-b(x)]
[(\vec{j}  (x+\zeta r\varepsilon)-\vec{j}(x))(b(x+\zeta r\varepsilon)-b(x))]d\sigma(\zeta)dr\\
=&\f12\pi\int_{0}^{\infty}r^{3}\varphi'( r)dr\int_{\partial B }\f{\zeta \cdot[\vec{j}(x+ \zeta r\varepsilon)-\vec{j}(x)]
[  b(x+\zeta r\varepsilon)-b(x) ]^{2} \f{d\sigma(\zeta)}{4\pi}}{r\varepsilon }\\
&-\pi\int_{0}^{\infty}r^{3}\varphi'( r)dr\int_{\partial B }\f{\zeta \cdot[b(x+ \zeta r\varepsilon)-b(x)]
[(\vec{j}(x+\zeta r\varepsilon)-\vec{j}(x))(b(x+\zeta r\varepsilon)-b(x))] \f{d\sigma(\zeta)}{4\pi}}{r\varepsilon }.
\ea\ee
It follows from
  integration by parts that
\be\ba\label{key3}
\int_{0}^{\infty}r^{3}\varphi'( r)dr= -3\int_{0}^{\infty}r^{2}\varphi ( r)dr
= -\f{3}{4\pi}\int_{\mathbb{R}^{3}} \varphi ( \ell)d\ell
= -\f{3}{4\pi}.
\ea\ee
Substituting \eqref{key3} into \eqref{key2}, one has
$$\ba
&D(v,\vec{j} )\\=&\lim_{\varepsilon\rightarrow0}D(v,\vec{j};\varepsilon)\\ =&\f\pi2\pi\int_{0}^{\infty}r^{3}\varphi'( r)dr\lim_{\varepsilon\rightarrow0}\int_{\partial B }\f{\zeta \cdot[\vec{j}(x+ \zeta r\varepsilon)-\vec{j}(x)]
[b(x+\zeta r\varepsilon)-b(x)]^{2}\f{d\sigma(\zeta)}{4\pi}}{r\varepsilon }\\&-\pi\int_{0}^{\infty}r^{3}\varphi'( r)dr\lim_{\varepsilon\rightarrow0}\int_{\partial B }\f{\zeta \cdot[b(x+ \zeta r\varepsilon)-b(x)]
[(\vec{j}(x+\zeta r\varepsilon)-\vec{j}(x))(b(x+\zeta r\varepsilon)-b(x))] \f{d\sigma(\zeta)}{4\pi}}{r\varepsilon }\\
=&\f38S_{1}(\vec{j},b,b)-\f{3}{4} S_{2}(b,\vec{j},b).
\ea$$
Thus, we conclude the Yaglom type relation \eqref{EMHDYaglom4/3law}.
 \end{proof}
 We will provided a slightly different approach to \eqref{2.15} as follows.
 \begin{proof}[Alternative proof of \eqref{2.15}]
It is clear that
$$\ba &\nabla\times[(\nabla\times b)\times b]\cdot b^{\varepsilon}+ \nabla\times[(\nabla\times b)\times b]^{\varepsilon}\cdot b\\
= &\f12\{2 \nabla\times[(\nabla\times b)\times b]\cdot b^{\varepsilon}+2 \nabla\times[(\nabla\times b)\times b]^{\varepsilon}\cdot b\}.
\ea$$
Thanks to \eqref{non1} and \eqref{non2}, we observe that
$$\ba& 2 \nabla\times[(\nabla\times b)\times b]\cdot b^{\varepsilon}]\\=&[\text{div}(b\otimes\vec{j})- \text{div}(\vec{j}\otimes b)]\cdot b^{\varepsilon}+\nabla\times[\text{div}(b\otimes b)]\cdot b^{\varepsilon}\\=&
[\text{div}(b\otimes\vec{j})- \text{div}(\vec{j}\otimes b)]\cdot b^{\varepsilon}+\text{div}([\text{div}(b\otimes b)]\times b^{\varepsilon})+[\text{div}(b\otimes b)]\cdot \vec{j}^{\varepsilon}.
\ea$$
and
$$\ba
&2 \nabla\times[(\nabla\times b)\times b]^{\varepsilon}\cdot b\\=&[\text{div}(b\otimes\vec{j})- \text{div}(\vec{j}\otimes b)]^{\varepsilon}\cdot b+\nabla\times[\text{div}(b\otimes b)]^{\varepsilon}\cdot b\\
=&[\text{div}(b\otimes\vec{j})- \text{div}(\vec{j}\otimes b)]^{\varepsilon}\cdot b+\text{div}([\text{div}(b\otimes b)]^{\varepsilon}\times b)+[\text{div}(b\otimes b)]^{\varepsilon} \cdot \vec{j}.
\ea$$
As a consequence, we get
\be\ba\label{key0}
&\nabla\times[(\nabla\times b)\times b]\cdot b^{\varepsilon}+ \nabla\times[(\nabla\times b)\times b]^{\varepsilon}\cdot b\\
=&\f12\{\text{div}([\text{div}(b\otimes b)]^{\varepsilon}\times b)+\text{div}([\text{div}(b\otimes b)]\times b^{\varepsilon})\}\\&+\f12\{\text{div}(b\otimes\vec{j})\cdot b^{\varepsilon}+[\text{div}(b\otimes b)]\cdot \vec{j}^{\varepsilon}+\text{div}(b\otimes\vec{j})^{\varepsilon}\cdot b+[\text{div}(b\otimes b)]^{\varepsilon} \cdot \vec{j}\\&- \text{div}(\vec{j}\otimes b)^{\varepsilon}\cdot b- \text{div}(\vec{j}\otimes b)\cdot b^{\varepsilon}\}.
\ea\ee
By means of this and  $\partial_{t}(bb^{\varepsilon})+\nabla\times[(\nabla\times b)\times b]\cdot b^{\varepsilon}+ \nabla\times[(\nabla\times b)\times b]^{\varepsilon}\cdot b=0$, one immediately gets \eqref{2.15}. Though this method may be easy, we  actually  obtain the first proof at the earliest when preparing this manuscript.
\end{proof}

We invoke the the dissipation term  \eqref{drdissipationterm} in Theorem  \ref{the1.1} to get  new energy conservation criterion of the EMHD equtions.
\begin{proof}[Proof of Corollary \ref{coro1.6}]
It follows from the H\"older inequality  that
$$
\int_{\mathbb{T}^{3}} |D_{\varepsilon}( v,\omega)|dx\leq \int_{\mathbb{T}^{3}}|\nabla\varphi_{\varepsilon}(\ell)|d\ell\B(\int_{\mathbb{T}^{3}}|\delta b(\ell)|^{\f{9}{2}}dx\B)^{\f49}\B(\int_{\mathbb{T}^{3}} |\delta \vec{j}(\ell)|^{\f{9}{5}} dx\B)^{\f{5}{9}}.
$$
In the light  of $D(b,\vec{j};\varepsilon)$ in  \eqref{drdissipationterm}, we get
 $$\ba
\int_{\mathbb{T}^{3}} |D(b,\vec{j};\varepsilon)|dx\leq & \int_{\mathbb{T}^{3}}|\nabla\varphi_{\varepsilon}(\ell)|d\ell\B(\int_{\mathbb{T}^{3}}|\delta b(\ell)|^{\f{9}{2}}dx\B)^{\f49}\B(\int_{\mathbb{T}^{3}} |\delta \vec{j}(\ell)|^{\f{9}{5}} dx\B)^{\f{5}{9}}\\
\leq&   \int_{\mathbb{T}^{3}}|\nabla\varphi_{\varepsilon}(\ell)| C(t)^{\f{2}{r_{1}}+\f{1}{r_{2}}}|\ell|^{2\alpha+\beta}\sigma(\ell)d\ell.
\ea$$
We conclude by performing a time integration and changing variable that
$$\ba
 \int_{0}^{T}\int_{\mathbb{T}^{3}} |D_{\varepsilon}( v,\omega)|dxdt \leq&
\int_{0}^{T}C(t)^{\f{2}{r_{1}}+\f{1}{r_{2}}}dt\int_{\mathbb{T}^{3}}
|\nabla\varphi_{\varepsilon}(\ell)|
|\ell|^{2\alpha+\beta}\sigma (\ell)d\ell\\
\leq&C \varepsilon^{2\alpha+\beta-1}\int_{|\xi|<1}|\nabla\varphi (\xi)
||\xi|^{2\alpha+\beta}\sigma (\varepsilon\xi)d\xi.
\ea$$
This leads to the desired result.
\end{proof}
\subsection{Exact relation of Magnetic helicity in the EMHD equations}
 \begin{proof}[Proof of Theorem \ref{the1.2}]
 With the help of \eqref{identity}, we  rewrite \eqref{magneticpotentialeq} as
$$ A_{t}+ (\nabla\times b)\times b   +\nabla \pi=A_{t}+\vec{j}\times b   +\nabla \pi= A_{t}+\text{div}(b\otimes b)  +\nabla (-\f{1}{2}|b|^{2}+\pi)=0.$$
Abusing notation slightly, we obtain
\be\label{2.20}
A_{t}+\text{div}(b\otimes b)  +\nabla  \pi=0.\ee
According to \eqref{2.20} and \eqref{eq2}, we know that
 \be\ba\label{2.21}
&A_{t}^{\varepsilon}b+A_{t}b^{\varepsilon}+ b_{t}^{\varepsilon}A+ b_{t}A^{\varepsilon}+\text{div}(b\otimes b)^{\varepsilon}b+\text{div}(b\otimes b)b^{\varepsilon}  +\nabla \pi^{\varepsilon}b\\&+\nabla \pi b^{\varepsilon}+\nabla\times[\text{div}(b\otimes b)]^{\varepsilon}A+\nabla\times[\text{div}(b\otimes b)]A^{\varepsilon}=0.
\ea\ee
From $ \eqref{VI}_{3}$, one arrives at
\be\label{h1}\ba
A\cdot\{  \nabla\times[\text{div}(b\otimes b)]^{\varepsilon}\}=&\text{div}([\text{div}(b\otimes b)]^{\varepsilon}\times A)+[\text{div}(b\otimes b)]^{\varepsilon}\cdot(\nabla\times  A )
\\=&\text{div}([\text{div}(b\otimes b)]^{\varepsilon}\times A)+[\text{div}(b\otimes b)]^{\varepsilon}\cdot b
\ea\ee
and
\be\label{h2}\ba
A^{\varepsilon}\cdot\{  \nabla\times[\text{div}(b\otimes b)]\}=&\text{div}([\text{div}(b\otimes b)]\times A^{\varepsilon})+[\text{div}(b\otimes b)]\cdot(\nabla\times A^{\varepsilon})
\\=&\text{div}([\text{div}(b\otimes b)]\times A^{\varepsilon})+[\text{div}(b\otimes b)]\cdot b^{\varepsilon}.
\ea\ee
Substituting this into \eqref{2.21}, we further deduce that
$$\ba
&b_{t}^{\varepsilon}A+ b_{t}A^{\varepsilon}+
A_{t}^{\varepsilon}b+A_{t}b^{\varepsilon}+ \text{div}([\text{div}(b\otimes b)]^{\varepsilon}\times A)
 +\text{div}([\text{div}(b\otimes b)]\times A^{\varepsilon}) +\\&  +\text{div}[ \pi^{\varepsilon}b+ \pi b^{\varepsilon}]=-2[ \text{div}(b\otimes b)^{\varepsilon}b+\text{div}(b\otimes b)b^{\varepsilon}].
\ea$$
An easy computation leads to that
\be\ba\label{h3}
\text{div}(b\otimes b)^{\varepsilon}b+\text{div}(b\otimes b)b^{\varepsilon}=&\partial_{k}(b_{k}  b_{i})^{\varepsilon}b_{i}+\partial_{k}(b_{k}  b_{i})b_{i}^{\varepsilon}\\=&
\partial_{k}(b_{k}  b_{i}b_{i}^{\varepsilon})+\partial_{k}(b_{k}  b_{i})^{\varepsilon}b_{i}-(b_{k}  b_{i})\partial_{k}b_{i}^{\varepsilon},
\ea\ee
which helps us to get
 $$\ba
&\f{(b^{\varepsilon}A)_{t}+ (bA^{\varepsilon})_{t}}{2} + \f12 \text{div}([\text{div}(b\otimes b)]^{\varepsilon}\times A)
 +\f12\text{div}([\text{div}(b\otimes b)]\times A^{\varepsilon}) \\&  +\f12\text{div}[ \Pi^{\varepsilon}b+ \Pi b^{\varepsilon}]+\f12 \partial_{k}(b_{k}  b_{i}b_{i}^{\varepsilon})=-[\partial_{k}(b_{k}  b_{i})^{\varepsilon}b_{i}-(b_{k}  b_{i})\partial_{k}
 b_{i}^{\varepsilon}].
\ea$$
Next, we show that we can pass to the limit in the above equation.
Indeed, the Sobolev embedding together with \eqref{the1.2c}    guarantee that $b\in L^{\infty}(0,T;L^{3}(\mathbb{T}^{3}))$.
The pressure equation in EMHD  \eqref{magneticpotentialeq}  is   determined   by
  $$
-\Delta\pi=\text{div div}(b\otimes b),
$$
  which means that
$$\ba
\|\pi\|_{L^{\f{3}{2}}(0,T;L^{\f{3}{2}}(\mathbb{T}^{3}))} \leq& C
\|b\|^{2}_{L^{3}(0,T;L^{3}(\mathbb{T}^{3}))}.
\ea$$
With this in hand, we are in a position to repeat the previous argument to prove the first part of this theorem. It is enough to get \eqref{EMHDMHYaglom4/3law}.
Following the path of \eqref{2.26}, we conclude that
\be\label{2.22}\ba
\int_{\mathbb{T}^{3}}\nabla\varphi_{\varepsilon}(\ell)\cdot\delta b(\ell)|\delta  b(\ell)|^{2}d\ell
=& \partial_{k}\B(b_{k}(  b_{i}^{2} )^{\varepsilon}-(b_{k} b_{i}^{2} )^{\varepsilon}\B)+2\partial_{k}(b_{k} b_{i}  )^{\varepsilon} b_{i}   -2b_{k}\partial_{k}  b_{i} ^{\varepsilon}b_{i}.
\ea\ee
 The derivation in \eqref{key2} and \eqref{key3}
entail that
$$\ba
D_{mh}(b,\varepsilon )=&\lim_{\varepsilon\rightarrow0}D_{\varepsilon}(\theta,v)\\
 =&-2\pi\int_{0}^{\infty}r^{3}\varphi'( r)dr\lim_{\varepsilon\rightarrow0}\int_{\partial B }\f{\zeta \cdot[b(x+ \zeta r\varepsilon)-b(x)]
[b(x+\zeta r\varepsilon)-b(x)]^{2}\f{d\sigma(\zeta)}{4\pi}}{r\varepsilon }\\
=&-\f{3}{2}S(b,b,b),
\ea$$
where the definition of $S(b,b,b)$ was used.

This  achieves the proof of this theorem.
 \end{proof}

\section{Four-thirds law in   Hall
magnetohydrodynamic equations}
Two $4/3$ laws   for the dissipation rates of energy and  magnetic helicity in the Hall
magnetohydrodynamic equations are established in this section.

\subsection{Exact relationship of energy in HMHD equations}
\begin{proof}[Proof of Theorem \ref{the1.4}]
 Similar to the derivation of \eqref{2.7}, we derive from \eqref{hallMHD} that
\be\ba\label{3.1}
&(u_{i}^{\varepsilon}u_{i})_{t} +(b_{i}b_{i}^{\varepsilon})_{t}
+\partial_{k}(u_{k}  u_{i})^{\varepsilon}u_{i}+\partial_{k}(u_{k}  u_{i})u_{i}^{\varepsilon}
-\partial_{k}(b_{k}  b_{i})^{\varepsilon}u_{i}-\partial_{k}(b_{k}  b_{i})u_{i}^{\varepsilon}
\\&+ \partial_{i}\Pi ^{\varepsilon}u_{i}+ \partial_{i}\Pi u_{i}^{\varepsilon}
+\partial_{k}(u_{k}  b_{i})^{\varepsilon}b_{i}  +\partial_{k}( u_{k}  b_{i})b_{i}^{\varepsilon}  -\partial_{k}(b_{k}  u_{i})^{\varepsilon}b_{i}-\partial_{k}(b_{k}  u_{i})b_{i}^{\varepsilon}\\&+[ \nabla\times[(\nabla\times b)\times b]\cdot b^{\varepsilon}+ \nabla\times[(\nabla\times b)\times b]^{\varepsilon}\cdot b]=0.
\ea\ee
After a few computations, we have
\begin{align}
&\partial_{k}(u_{k}  u_{i})^{\varepsilon}u_{i}+\partial_{k}(u_{k}  u_{i})u_{i}^{\varepsilon} = \partial_{k}(u_{k}  u_{i}u_{i}^{\varepsilon})+\partial_{k}(u_{k}  u_{i})^{\varepsilon}u_{i}-
 u_{k}  u_{i}\partial_{k}u_{i}^{\varepsilon}\nonumber \\
 &\partial_{k}(b_{k}  b_{i})^{\varepsilon}u_{i}+\partial_{k}(b_{k}  b_{i})u_{i}^{\varepsilon}=\partial_{k}(b_{k}  b_{i}u_{i}^{\varepsilon})+\partial_{k}(b_{k}  b_{i})^{\varepsilon}u_{i}-(b_{k}  b_{i})\partial_{k}u_{i}^{\varepsilon}\nonumber\\
 &\partial_{i}\Pi ^{\varepsilon}u_{i}+ \partial_{i}\Pi u_{i}^{\varepsilon}=\partial_{i}(\Pi^{\varepsilon}u_{i}+ \Pi u_{i}^{\varepsilon})
\nonumber \\
 &\partial_{k}(u_{k}  b_{i})^{\varepsilon}b_{i}  +\partial_{k}( u_{k}  b_{i})b_{i}^{\varepsilon}=\partial_{k}( u_{k}  b_{i}b_{i}^{\varepsilon})+\partial_{k}(u_{k}  b_{i})^{\varepsilon}b_{i}-( u_{k}  b_{i})\partial_{k}b_{i}^{\varepsilon}\nonumber\\
      &\partial_{k}(b_{k}  u_{i})^{\varepsilon}b_{i}+\partial_{k}(b_{k}  u_{i})b_{i}^{\varepsilon}=\partial_{k}(b_{k}  u_{i}b_{i}^{\varepsilon})+\partial_{k}(b_{k}  u_{i})^{\varepsilon}b_{i}-(b_{k}  u_{i})\partial_{k}b_{i}^{\varepsilon}.\label{3.2}
 \end{align}
Employing \eqref{key0} and \eqref{2.14}, we see that
\begin{align}
&\nabla\times[(\nabla\times b)\times b]\cdot b^{\varepsilon}+ \nabla\times[(\nabla\times b)\times b]^{\varepsilon}\cdot b\nonumber\\ =&\f12\{\text{div}([\text{div}(b\otimes b)]^{\varepsilon}\times b)+\text{div}([\text{div}(b\otimes b)]\times b^{\varepsilon})\}\nonumber\\&+\f12\{\text{div}(b\otimes\vec{j})\cdot b^{\varepsilon}+[\text{div}(b\otimes b)]\cdot \vec{j}^{\varepsilon}+\text{div}(b\otimes\vec{j})^{\varepsilon}\cdot b+[\text{div}(b\otimes b)]^{\varepsilon} \cdot \vec{j}\nonumber\\&- \text{div}(\vec{j}\otimes b)^{\varepsilon}\cdot b- \text{div}(\vec{j}\otimes b)\cdot b^{\varepsilon}\}\nonumber\\
=&\f12 \text{div}([\text{div}(b\otimes b)]^{\varepsilon}\times b)+\f12\text{div}([\text{div}(b\otimes b)]\times b^{\varepsilon}) \nonumber\\&+\f12\text{div}[b (b\cdot j^{\varepsilon})] +\f12\text{div}[b (b^{\varepsilon}\cdot j)] +\f12[\partial_{k}(b_{k}b_{i})^{\varepsilon}\vec{j}_{i}-(b_{k}b_{i}) \partial_{k}j_{i}^{\varepsilon} +\partial_{k}(b_{k}\vec{j}_{i})^{\varepsilon}b_{i}-(b_{k}j_{i}) \partial_{k}b _{i}^{\varepsilon}]\nonumber\\&- \f12 \text{div}(\vec{j} b\cdot b^{\varepsilon})-\f12[\partial_{k}(\vec{j}_{k}b_{i})^{\varepsilon}b_{i}
-(\vec{j}_{k}b_{i})\partial_{k}b_{i}^{\varepsilon}].\label{3.3}
 \end{align}
Inserting \eqref{3.1} and \eqref{3.2} into\eqref{3.3} , we observe that
$$\ba
&\f{(u_{i}^{\varepsilon}u_{i})_{t} +(b_{i}b_{i}^{\varepsilon})_{t}}{2}
+\f{\partial_{k}(u_{k}  u_{i}u_{i}^{\varepsilon})
-\partial_{k}(b_{k}  b_{i}u_{i}^{\varepsilon})+\partial_{k}( u_{k}  b_{i}b_{i}^{\varepsilon})-\partial_{k}( u_{k}  b_{i}b_{i}^{\varepsilon}) }{2}
\\ &+\f{\partial_{i}(\Pi^{\varepsilon}u_{i}+ \Pi u_{i}^{\varepsilon})}{2}
+\f14 \text{div}([\text{div}(b\otimes b)]^{\varepsilon}\times b)+\f14\text{div}([\text{div}(b\otimes b)]\times b^{\varepsilon}) \\ &+\f14\text{div}[b (b\cdot j^{\varepsilon})]+\f14\text{div}[b (b^{\varepsilon}\cdot j)]- \f14 \text{div}(\vec{j} b\cdot b^{\varepsilon})\\
=&-\f12[\partial_{k}(u_{k}  u_{i})^{\varepsilon}u_{i}-
 u_{k}  u_{i}\partial_{k}u_{i}^{\varepsilon}]+\f12[\partial_{k}(b_{k}  b_{i})^{\varepsilon}u_{i}-(b_{k}  b_{i})\partial_{k}u_{i}^{\varepsilon}+\partial_{k}(b_{k}  u_{i})^{\varepsilon}b_{i}-(b_{k}  u_{i})\partial_{k}b_{i}^{\varepsilon}]\\
 &+\f14[\partial_{k}(\vec{j}_{k}b_{i})^{\varepsilon}b_{i}
-(\vec{j}_{k}b_{i})\partial_{k}b_{i}^{\varepsilon}]\\ &
-\f14[\partial_{k}(b_{k}b_{i})^{\varepsilon}\vec{j}_{i}-(b_{k}b_{i}) \partial_{k}j_{i}^{\varepsilon} +\partial_{k}(b_{k}\vec{j}_{i})^{\varepsilon}b_{i}-(b_{k}j_{i}) \partial_{k}b _{i}^{\varepsilon}].
\ea$$
Exactly as the derivation of \eqref{2.26}, we discover that
$$
\ba
&\int_{\mathbb{T}^{3}}\nabla\varphi_{\varepsilon}(\ell)\cdot\delta u(\ell)|\delta  u(\ell)|^{2}d\ell
= \partial_{k}\B[u_{k}(  b_{i}^{2} )^{\varepsilon}-(u_{k} u_{i}^{2} )^{\varepsilon}\B]+2\partial_{k}(u_{k} u_{i}  )^{\varepsilon} u_{i}   -2u_{k}\partial_{k} u_{i} ^{\varepsilon}u_{i},\\
&\int_{\mathbb{T}^{3}}\nabla\varphi_{\varepsilon}(\ell)\cdot\delta b(\ell) | \delta  u(\ell)\cdot\delta b(\ell)|  d\ell\\
=&\partial_{k}\B[b_{k} (u_{i}v_{i})^{\varepsilon}-(b_{k}u_{i}b_{i})^{\varepsilon}\B]
+\partial_{k}(b_{k}b_{i})^{\varepsilon}u_{i}
-b_{k}b_{i}\partial_{k}u_{i}^{\varepsilon}
 +\partial_{k}(b_{k}u_{i})^{\varepsilon}b_{i}
-b_{k}u_{i}\partial_{k}b_{i}^{\varepsilon}, \\ &\int_{\mathbb{T}^{3}}\nabla\varphi_{\varepsilon}(\ell)\cdot\delta u(\ell)|\delta  b(\ell)\cdot\delta  b(\ell)| d\ell
= \partial_{k}\B[u_{k}( b_{i}^{2} )^{\varepsilon}-(\vec{j}_{k}b_{i}^{2} )^{\varepsilon}\B]+2\partial_{k}(u_{k} b_{i}  )^{\varepsilon} b_{i}-2u_{k}\partial_{k}  b_{i} ^{\varepsilon}b_{i}.
\ea $$
Recall \eqref{2.26} and \eqref{2.17}, one obtain
$$\int_{\mathbb{T}^{3}}\nabla\varphi_{\varepsilon}(\ell)\cdot\delta \vec{j}(\ell)|\delta  b(\ell)\cdot\delta  b(\ell)| d\ell
= \partial_{k}\B[\vec{j}_{k}( b_{i}^{2} )^{\varepsilon}-(\vec{j}_{k}b_{i}^{2} )^{\varepsilon}\B]+2\partial_{k}(\vec{j}_{k} b_{i}  )^{\varepsilon} b_{i}-2\vec{j}_{k}\partial_{k}  b_{i} ^{\varepsilon}b_{i},$$
and
\be\ba
&\int_{\mathbb{T}^{3}}\nabla\varphi_{\varepsilon}(\ell)\cdot\delta b(\ell) | \delta  j(\ell)\cdot\delta b(\ell)|  d\ell\\
 =&\partial_{k}\B[b_{k} (\vec{j}_{i}v_{i})^{\varepsilon}-(b_{k}\vec{j}_{i}b_{i})^{\varepsilon}\B]
+\partial_{k}(b_{k}b_{i})^{\varepsilon}\vec{j}_{i}
-b_{k}b_{i}\partial_{k}\vec{j}_{i}^{\varepsilon}
 +\partial_{k}(b_{k}\vec{j}_{i})^{\varepsilon}b_{i}
-b_{k}\vec{j}_{i}\partial_{k}b_{i}^{\varepsilon}.
\ea\ee
A similar procedure for \eqref{key2}, we write
 \begin{align}
&D(u,b,\vec{j};\varepsilon)\nonumber\\
=& -\f14\int_{\mathbb{T}^{3}}\nabla\varphi_{\varepsilon}(\ell)\cdot\delta u(\ell)|\delta  u(\ell)|^{2}d\ell -\f14\int_{\mathbb{T}^{3}}\nabla\varphi_{\varepsilon}(\ell)\cdot\delta b(\ell) | \delta  u(\ell)\cdot\delta b(\ell)|  d\ell\nonumber\\&+\f12 \int_{\mathbb{T}^{3}}\nabla\varphi_{\varepsilon}(\ell)\cdot\delta u(\ell)|\delta  b(\ell)\cdot\delta  b(\ell)| d\ell\nonumber\\&+ \f18\int_{\mathbb{T}^{3}}\nabla\varphi_{\varepsilon}(\ell)\cdot\delta\vec{ j}(\ell) |\delta  b(\ell)|^{2} d\ell-\f14\int_{\mathbb{T}^{3}}\nabla\varphi_{\varepsilon}(\ell)\cdot\delta b(\ell) | \delta  \vec{j}(\ell)\cdot\delta b(\ell)|  d\ell\nonumber\\=&-\pi\int_{0}^{\infty}r^{3}\varphi'( r)dr\int_{\partial B }\f{\zeta \cdot[b(x+ \zeta r\varepsilon)-b(x)]
[b(x+\zeta r\varepsilon)-b(x)]^{2}\f{d\sigma(\zeta)}{4\pi}}{r\varepsilon }\nonumber\\&-\pi\int_{0}^{\infty}r^{3}\varphi'( r)dr\int_{\partial B }\f{\zeta \cdot[b(x+ \zeta r\varepsilon)-b(x)]
[b(x+\zeta r\varepsilon)-b(x)]^{2}\f{d\sigma(\zeta)}{4\pi}}{r\varepsilon }\nonumber\\&+2\pi\int_{0}^{\infty}r^{3}\varphi'( r)dr\int_{\partial B }\f{\zeta \cdot[b(x+ \zeta r\varepsilon)-b(x)]
[b(x+\zeta r\varepsilon)-b(x)]^{2}\f{d\sigma(\zeta)}{4\pi}}{r\varepsilon }\nonumber\\&+\f12\pi\int_{0}^{\infty}r^{3}\varphi'( r)dr\int_{\partial B }\f{\zeta \cdot[\vec{j}(x+ \zeta r\varepsilon)-\vec{j}(x)]
[  b(x+\zeta r\varepsilon)-b(x) ]^{2} \f{d\sigma(\zeta)}{4\pi}}{r\varepsilon }\nonumber\\
&-\pi\int_{0}^{\infty}r^{3}\varphi'( r)dr\int_{\partial B }\f{\zeta \cdot[b(x+ \zeta r\varepsilon)-b(x)]
[(\vec{j}(x+\zeta r\varepsilon)-\vec{j}(x))(b(x+\zeta r\varepsilon)-b(x))] \f{d\sigma(\zeta)}{4\pi}}{r\varepsilon}.\nonumber
 \end{align}
 It follows from  \eqref{key3} and the definition of $S_{i}$ that
$$\ba  D(u,b,\vec{j}) =&\lim_{\varepsilon\rightarrow0}D(v,\vec{j};\varepsilon)
\\=&-\f34 S_{3}(u,u,u)-\f34 S_{4}(u,b,b)+\f{3}{2}S_{5}(b,u,b)+\f38S_{1}(\vec{j},b,b)-\f{3}{4} S_{2}(b,\vec{j},b).
\ea $$
This leads to the desired results.
\end{proof}
\subsection{Exact relationship of  Magnetic helicity in the HMHD system}

\begin{proof}[Proof of Theorem \ref{the1.5}]
In the light of   $\eqref{VI}_{2}$  and \eqref{non2}, we obtain an equivalent alternative formulation  of $\eqref{hallMHD}_{2}$
\be\label{035}
b_{t}+\nabla\times(b\times u)+\nabla\times[\text{div}(b\otimes b)]=0.
\ee
By arguing as was done to obtain \eqref{2.20},
abusing notation slightly, we conclude by \eqref{hmpotentialeq} that
\be\label{3.6}
A_{t}-u\times b+\text{div}(b\otimes b)  +\nabla  \pi=0.\ee
Thanks  to\eqref{035} and \eqref{3.6}, we arrive at
 \be\ba\label{3.7}
&(A^{\varepsilon}b)_{t}+(Ab^{\varepsilon})_{t} -(u\times b)b^{\varepsilon}-(u\times b)^{\varepsilon}b-\nabla\times(u\times b)^{\varepsilon}\cdot A-\nabla\times(u\times b)\cdot A^{\varepsilon}+
\text{div}(b\otimes b)^{\varepsilon}b\\&+\text{div}(b\otimes b)b^{\varepsilon}  +\nabla \pi^{\varepsilon}b +\nabla \pi b^{\varepsilon}+\nabla\times[\text{div}(b\otimes b)]^{\varepsilon}A+\nabla\times[\text{div}(b\otimes b)]A^{\varepsilon}=0.
\ea\ee
It follows from $ \eqref{VI}_{3}$ that
$$\ba
A\cdot[\nabla\times(u\times b)^{\varepsilon}] =&\text{div}[(u\times b)^{\varepsilon}\times A]+(u\times b)^{\varepsilon}\cdot(\nabla\times A)
\\=&\text{div}[(u\times b)^{\varepsilon}\times A]+(u\times b)^{\varepsilon}\cdot b.
\ea$$
and
$$\ba
A^{\varepsilon}\cdot[\nabla\times(u\times b)]
 =&\text{div}[(u\times b)\times A^{\varepsilon}]+(u\times b)\cdot b^{\varepsilon}.
\ea$$
Inserting this \eqref{h1} and \eqref{h2} into \eqref{3.7},
we remark that
$$\ba
&(A^{\varepsilon}b)_{t}+(Ab^{\varepsilon})_{t}-\text{div}[(u\times b)^{\varepsilon}\times A]-(u\times b)^{\varepsilon}\cdot b-\text{div}[(u\times b)\times A^{\varepsilon}]-(u\times b)\cdot b^{\varepsilon}\\&+ \text{div}([\text{div}(b\otimes b)]^{\varepsilon}\times A)
 +\text{div}([\text{div}(b\otimes b)]\times A^{\varepsilon}) +\text{div}[ \pi^{\varepsilon}b+ \pi b^{\varepsilon}]\\
 =&-2[ \text{div}(b\otimes b)^{\varepsilon}b+\text{div}(b\otimes b)b^{\varepsilon}].
\ea$$
which together with \eqref{h3} implies that
 $$\ba
&\f{(b^{\varepsilon}A)_{t}+ (bA^{\varepsilon})_{t}}{2}+\f12\text{div}[(u\times b)^{\varepsilon}\times \vec{A}]+\f12\text{div}[(u\times b)\times \vec{A}^{\varepsilon}] + \f12 \text{div}([\text{div}(b\otimes b)]^{\varepsilon}\times \vec{A})\\&
 +\f12\text{div}([\text{div}(b\otimes b)]\times \vec{A}^{\varepsilon}) +\f12\text{div}[ \pi^{\varepsilon}b+ \pi b^{\varepsilon}]+\f12 \partial_{k}(b_{k}  b_{i}b_{i}^{\varepsilon})\\=&-[\partial_{k}(b_{k}  b_{i})^{\varepsilon}b_{i}-(b_{k}  b_{i})\partial_{k}
 b_{i}^{\varepsilon}]  +\f12(u\times b)^{\varepsilon}\cdot b +\f12(u\times b)\cdot b^{\varepsilon}.
\ea$$
Just notice
$u,b\in L^{3}(0,T;L^{3}(\mathbb{T}^{3})) $ ensures that the limit
$\f12(u\times b)^{\varepsilon}\cdot b +\f12(u\times b)\cdot b^{\varepsilon}$ is zero as $\varepsilon\rightarrow0.$ The rest proof is the same as the argument in previous  subsection. We omit the detail here.
This completes the proof.
\end{proof}
\section{Conclusion}
The first 4/3 relation for
mixed moments of the velocity and temperature fields was due to Yaglom in \cite{[Yaglom]}. This type  four-thirds law   exists in a large number of  turbulence models (see e.g. \cite{[AOAZ],[AB],[Chkhetiani1],[Chkhetiani],
 [GPP],[PP1],[PP2],[Podesta],[PFS],[YRS],[WWY1],[Galtier1],[FGSMA],
 [HVLFM],[Galtier]}). The nonlinear terms in these models are almost all in terms of convection type rather than Hall type.
Making full use structure of the Hall term, we present four 4/3 laws of  the dissipation rates of energy of energy and  magnetic helicity in electron and Hall
magnetohydrodynamic equations. It is worth pointing out that the  four-thirds relation \eqref{HMHDYaglom4/3law} or  \eqref{EMHDYaglom4/3law} is different from the known results   in \cite{[Galtier1],[Galtier],[BG]}. Though 2/15 law of magnetic helicity in EMHD system was derived by Chkhetiani in \cite{[Chkhetiani2]} and other exact relations for the magnetic helicity in the HMHD equations were  presented by Banerjee-Galtier in \cite{[BG]},
the 4/3 law obtained here
 does not exist in the known literature. Hence, Theorem \ref{the1.1} and \ref{the1.5} are the first results in this direction. The exact laws in Theorem  \ref{the1.1}-\ref{the1.5} can be viewed as the generalized Yaglom type law.

 It is useful to understand  the difference between the standard MHD equations involving the  convection terms  the HMHD equations containing   Hall term.
 It should be remarked that 4/3 relation   \eqref{EMHDYaglom4/3law}  of the energy in the EMHD equations is similar to the one for the helicity in the Euler equations in \cite{[WWY1]} and 4/3 law  \eqref{EMHDMHYaglom4/3law}   for the  magnetic helicity  is the same as the one for the energy  in the Euler equations in \cite{[DR]}.
 Moreover, the results closely related to generalized  Onsager conjecture that the critical regularity of weak solutions ensures that the
 conserved law  is valid (see Corollary \ref{coro1.6}). The  Onsager conjecture and its  generalized version can be found \cite{[CCFS],[BGSTW],[BT],[CKS],[Chae1],
[CET],[DKL],[De Rosa],[DR],[DE],
[Eyink0],[WZ]}.

It is an interesting question to derive \eqref{EMHDMHYaglom4/3law} and \eqref{HMHDMHYaglom4/3law} via the corresponding K\'arm\'an-Howarth type equations.
 \section*{Acknowledgement}

 Wang was partially supported by  the National Natural
 Science Foundation of China under grant (No. 11971446, No. 12071113   and  No.  11601492) and  sponsored by Natural Science Foundation of Henan.


\begin{thebibliography}{00}











 \bibitem{[AB]}
A. Alexakis and L. Biferale,
Cascades and transitions in turbulent flows. Physics Reports, 767 (2018), 1--101.



 \bibitem{[AOAZ]}
R. Antonia, M. Ould-Rouis, F. Anselmet and Y. Zhu,
Analogy between predictions of Kolmogorov
and Yaglom.  J. Fluid Mech.,  332 (1997),   395--409.


 \bibitem{[BG]}
S. Banerjee and S. Galtier
Chiral exact relations for helicities in Hall magnetohydrodynamic turbulence.
Physical Review E. 93 (2016), 033120.

 \bibitem{[BGSTW]}
 C. Bardos, P. Gwiazda, A. \'Swierczewska-Gwiazda, E. S. Titi and   E. Wiedemann,  Onsager's conjecture in bounded domains for the conservation of entropy and other companion laws. Proc. R. Soc. A, 475 (2019),   18 pp.

\bibitem{[Biskamp99]}
D. Biskamp, E. Schwarz, F. Zeiler, A. Celani and J. F. Drake Electron magnetohydrodynamic turbulence. Phys, Plasmas 6 (1999), P. 751--758.

 \bibitem{[BT]}
D. W.  Boutros and E. S. Titi, Onsager's conjecture for subgrid scale $\alpha$-models of turbulence. Phys. D 443 (2023), Paper No. 133553, 23 pp.

\bibitem{[CKS]}
R. E. Caflisch, I. Klapper and G. Steele.  Remarks on singularities, dimension and energy dissipation for ideal hydrodynamics and MHD. Commun. Math. Phys. 184 (1997) 443--55.


\bibitem{[Chae1]} D. Chae,
On the Conserved Quantities for the Weak Solutions
of the Euler Equations and the Quasi-geostrophic
Equations
Commun. Math. Phys. 266 (2006), 197--210.



\bibitem{[CCFS]}
A. Cheskidov,  P. Constantin, S. Friedlander and R. Shvydkoy, Energy conservation and Onsager's conjecture for the Euler equations. Nonlinearity, 21 (2008), 1233--52.
\bibitem{[Chkhetiani]}
O. Chkhetiani, On the third moments in helical turbulence. JETP Lett. 63 (1996) , 808--812.
\bibitem{[Chkhetiani2]}O. Chkhetiani,
On triple correlations in isotropic electronic
magnetohydrodynamic turbulence. JETP Lett. 69 (1999), 664-668.

\bibitem{[Chkhetiani1]}O.  Chkhetiani, On the local structure of helical turbulence.  Doklady Physics.  53 (2008),  513--516.


\bibitem{[CET]} P. Constantin, W. E, and E. S. Titi, Onsager's conjecture on the energy conservation for solutions of Euler's equation. Commun. Math. Phys. 165   (1994), 207--209.



\bibitem{[DKL]}
M. Dai,  J. Krol and H. Liu, On uniqueness and helicity conservation of weak solutions to the electron-MHD system. J. Math. Fluid Mech. 24 (2022),   17 pp.



\bibitem{[De Rosa]}
L. De Rosa, On the helicity conservation for the incompressible Euler equations, Proc. Amer.
Math. Soc. 148 (2020),  2969--2979.
\bibitem{[DR]}
J. Duchon and R. Robert,  Inertial Energy Dissipation for Weak Solutions of Incompressible Euler and
Navier-Stokes Equations. Nonlinearity. 13 (2000), 249--255.
\bibitem{[DE]}
    T. Drivas and G. Eyink. An Onsager singularity theorem for turbulent solutions of compressible Euler
equations.  Commun. Math. Phys., 359 (2018), 733--763.


\bibitem{[Eyink0]}
G. Eyink, Energy dissipation without viscosity in ideal hydrodynamics I. Fourier
analysis and local energy transfer. Physica D: Nonlinear Phenomena 78(1994), 222--240.
\bibitem{[Eyink2]}G. Eyink,
Intermittency and anomalous scaling of passive scalars in any space dimension.
Physical Review E 54.2 (1996), 1497.

\bibitem{[Eyink1]}G. Eyink, Local 4/5-law and energy dissipation anomaly in turbulence. Nonlinearity 16( 2003),
137--145.

\bibitem{[FGSMA]}
R. Ferrand, S. Galtier, F. Sahraoui, R. Meyrand, N. Andr\'es and S. Banerjee,
On exact laws in incompressible Hall magnetohydrodynamic turbulence.
The Astrophysical Journal, 881(2019), 6pp.

\bibitem{[Frick03]} P. Frick, R. Stepanov and V. Nekrasov
Shell model of the magnetic field evolution under Hall effect. Magnetohydrodynamics, 39 (2003), 327--334.

\bibitem{[Frisch]}
U. Frisch, Turbulence. Cambridge University Press. 1995

\bibitem{[Galtier]}
S. Galtier,
On the origin of the energy dissipation anomaly in
(Hall) magnetohydrodynamics J. Phys. A: Math. Theor. 51 (2018) 205501



\bibitem{[Galtier1]}
S. Galtier, Introduction to Modern Magnetohydrodynamics. Cambridge:
Cambridge Univ. Press.  2016.


\bibitem{[GPP]}
T. Gomez, H. Politano and A. Pouquet, Exact relationship for third-order structure functions in helical flows,
Phys. Rev. E 61 (2000), 5321--5325.


\bibitem{[HVLFM]}
P. Hellinger, A. Verdini, S. Landi  , L. Franci  and L. Matteini,
von K\'arm\'an-Howarth Equation for Hall Magnetohydrodynamics: Hybrid Simulations. The Astrophysical Journal Letters, 857 (2018), 5pp.


\bibitem{[KCY]}
A. S. Kingsep,  K. V. Chukbar,   and V. V.  Yan'kov, in Reviews of Plasma
Physics, Vol. 16, ed. B. Kadomtsev (New York: Consultants Bureau), 1990.


\bibitem{[Kolmogorov]}
A. N. Kolmogorov, Dissipation of energy in the locally isotropic turbulence. Dokl. Adad.
Nauk SSSR 32(1941). English transl. Proc. R. Soc. Lond. A 434 (1991), 15--17.

\bibitem{[MY]}
 A. S.  Monin and A. M. Yaglom, Statistical Fluid Mechanics: Mechanics of Turbulence, Vol. 2.
The MIT Press. 1975

 \bibitem{[PP1]}
H. Politano and A.  Pouquet,   Dynamical length scales for turbulent magnetized flows. Geophys. Res. Lett. 25 (1998), 273.

 \bibitem{[PP2]}H. Politano and A.  Pouquet, von K\'arm\'an-Howarth equation for magnetohydrodynamics and its consequences on third-order longitudinal structure
 and correlation functions. Phys. Rev. E 57 (1998), R21.

\bibitem{[Podesta]}J. Podesta, Laws for third-order moments in homogeneous anisotropic incompressible magnetohydrodynamic
turbulence, J. Fluid Mech. 609 (2008), 171--194.

\bibitem{[PFS]} J. Podesta, M. Forman and C. Smith, Anisotropic form of third-order moments and relationship to the cascade
rate in axisymmetric magnetohydrodynamic turbulence, Physics of Plasmas, 14 (2007), 092305.

\bibitem{[Vainshtein73]} S. I. Vainshtein Strong plasma turbulence at helicon frequencies Sov. Phys.-JETP 37 (1973), 73-76.

  \bibitem{[WZ]}
Y. Wang and B. Zuo, Energy and cross-helicity conservation for the three-dimensional ideal MHD equations in a bounded domain. J. Differential Equations 268 (2020),   4079--4101.

\bibitem{[WWY1]}
Y. Wang, W. Wei and Y. Ye,
Yaglom's law and conserved quantity dissipation in turbulence
  arXiv:2301.10917v2.



 \bibitem{[Yaglom]}
A. M. Yaglom, Local Structure of the Temperature Field in a Turbulent Flow. Dokl. Akad. Nauk SSSR, 69 (1949), 743--746.
\bibitem{[YRS]}
T. A. Yousef, F. Rincon  and A. A. Schekochihin,
Exact scaling laws and the local structure of
isotropic magnetohydrodynamic turbulence.
J. Fluid Mech.,  575(2007),   111--120.












\end{thebibliography}
\end{document}